\documentclass[11pt]{article}
\usepackage{anyfontsize}
\usepackage{graphicx} 
\usepackage{amsmath,amsthm,amssymb,mathrsfs}

\newtheorem{theorem}{Theorem}
\newtheorem{proposition}{Proposition}
\newtheorem{lemma}{Lemma}

\newtheorem{assumption}{Assumption}
\newtheorem{remark}{Remark}
\newtheorem{example}{Example}

\begin{document}

\title{Projection Operator Stochastic Equations for Non-Markovian Quantum Systems Under Continuous Measurement-Based Feedback}
 
\author{Hendra I. Nurdin\\ School of Electrical Engineering and Telecommunications\\ University of New South Wales, Sydney, NSW 2052. \\
\texttt{Email: h.nurdin@unsw.edu.au}}
\date{}




\maketitle

\begin{abstract}
 Quantum Markov models have been successfully used to accurately model various physical quantum systems in fields such as quantum optics, optomechanics and superconducting circuits and they provide  the basis  for (measurement-based) quantum feedback control. However, the quantum Markov assumption is a strong one and it is not expected to hold for general quantum systems of interest. The projection operator approach is one approach that has been developed to model non-Markovian quantum systems by considering its embedding in a larger Markovian quantum system, but mainly in the context of quantum master equations for the dynamics of the unmonitored reduced quantum state of a quantum system. This approach was recently adapted for continuously measured non-Markovian quantum systems, which enables open-loop control but did not yet consider the presence of feedback of the stochastic measurement record, deriving non-Markovian SDEs for the evolution of the projected state of the Markovian embedding. This paper generalizes these stochastic equations to the setting  of stochastic feedback based on the continuous-measurement record and shows that the equations take the same form but that previously deterministic terms become stochastic ones which depend on the measurement record, as would be intuitively expected. The stochastic equations are obtained for a generalized class of measurements that includes continuous (possibly adaptive) homodyne and photon counting measurements.
\end{abstract}

\section{Introduction}
Quantum Markov models, as the quantum analogue of classical Markov models, have played a prominent role in the modeling of open quantum systems. These models yields the Lindblad quantum master equation for describing the (unconditional) reduced state of open quantum systems  and the quantum filtering equation (stochastic master equation) for describing the stochastic quantum evolution of an open quantum system undergoing continuous measurement in the Heisenberg (Schr\"{o}dinger  picture) \cite{BvHJ07,BvH08,WM10,NY17}. In a quantum Markov model the open quantum system interacts with an environment made of  quantum (white noise) fields that do not retain memory of the system's past states. The closed quantum system consisting of the open quantum system and its quantum field environment has a joint unitary evolution that is given by the solution of a quantum stochastic differential equation (QSDE) \cite{HP84,GC85,KRP92,GZ04}.   

Although quantum Markov models have found success as models for various physical systems found in fields such as quantum optics, optomechanics and superconducting circuits, the quantum Markov assumption is rather strong and is not anticipated to hold for a wide range open quantum systems \cite{BP07}. In particular, it may not be appropriate to model the environment as memoryless quantum white noises. Indeed, some studies in the literature have indicated that quantum noise in near-term quantum computers exhibit non-Markovian features, see, e.g., \cite{White20}. Therefore, it remains an important endeavour to look beyond quantum Markov models and develop the appropriate tools.

 One approach to modeling non-Markovian quantum systems is to model the environment in a compound way as consisting of another quantum system, called an {\em auxiliary  quantum system}, which may in turn be coupled to a quantum white noise environment; see \cite{Nurd23,Nurd25} and the references therein. The auxiliary acts a quantum memory subsystem whose quantum state can depend on past states of the open quantum system of interest, which will be referred to as the {\em principal quantum system}. An auxiliary that consists of a collection of quantum harmonic oscillators is often referred to as pseudomodes \cite{DBG01,Mascherpa20}. In this model of a principal quantum system with compound environment (auxiliary system together with quantum white noises), projection operator methods have been developed to study the dynamics of the principal quantum system \cite{BP07}\cite[\S 3]{Bacchini11}. It is based on using a linear projection superoperator $\mathcal{P}$  with certain properties  (see the definition herein  in \S \ref{sec:main-results}) that projects the joint density operator $\rho_{\rm sa}$ of the principal and auxiliary to a subspace of density operators such that the reduced state of the principal quantum system alone can be obtained from the projected states by tracing out the auxiliary. Within this projection operator framework, one can obtain the so-called Nakajima-Zwanzig quantum master equation \cite{BP07,Bacchini11} by eliminating the component of $\rho_{\rm sa}$ that belongs to the subspace orthogonal to $\mathcal{P}$ from either a Liouville-von Neumann equation (corresponding to Hamiltonian coupling between the principal and auxiliary) or from a Lindblad quantum master equation for the principal and auxiliary. Further assumptions lead to more simplified non-Markovian master equations such as the time-convolutionless non-Markovian quantum master equation. 

It was shown recently in \cite{Nurd25} that the projection operator approach and the Nakajima-Zwanzig equation can be adapted to non-Markovian principal quantum systems that are coupled to a probe and for which the probe is continuously measured. The projected density operator of the principal and auxiliary evolves stochastically according to a new non-Markovian stochastic differential equation (SDE) that includes an integral term over a stochastic two-time kernel which depends on the past of the projected  density operator. However, while the model considered in \cite{Nurd25} supports open-loop control, it does not include feedback based on the continuous-measurement record. In this work we show that the non-Markovian SDEs of \cite{Nurd25} can be generalized to the case when feedback based on the continuous measurement record is included. Moreover, this paper derives the results for a very general class of measurement, beyond the standard ones considered in \cite{Nurd25}.  

This paper is structured as follows. Section \ref{sec:non-Markov-model} sets up the modeling framework for non-Markovian quantum systems considered in this paper. Section \ref{sec:controlled-flows} reviews controlled QSDEs, controllled quantum flows and and the associated output equation under a very general class of measurements, following \cite{BvH08}. This is followed in Section \ref{sec:controlled-quantum-fitering} with an overview of the quantum filtering equation for controlled quantum systems.  Section \ref{sec:main-results} presents the main results of this paper. This is Theorem \ref{thm:non-Markovian-diffusion} for a non-Markovian quantum system in the diffusive case (i.e., under  generalized continuous homodyne measurements) and Theorem \ref{thm:non-Markovian-jump} for the pure jump  case (i.e., (i.e., under generalized continuous photon counting measurements) 
followed by a discussion on the inherent feedback loop structure of the non-Markovian SDEs presented in these theorems. Finally, Section \ref{sec:conclu} gives a conclusion for the paper and directions for future research.  \\

\noindent \textbf{Notation.} We will adopt the notation in \cite{Nurd25} with additional notations introduced in the text as required. $X^{\top}$ denotes the transpose of a matrix $X$, $X^{\dag}$ denotes the adjoint of a Hilbert space operator $X$ (the conjugate transpose when $X$ is a matrix). $I_n$ will denote an $n \times n$ identity matrix and $I$ can denote either an identity matrix (whose dimension can be inferred from the context), an identity map or an identity operator. $\mathrm{Tr}$ denotes the trace of a matrix or an operator. A signal (function of time) will be denoted by $V_{\cdot}$ where the subscript $\cdot$ is a placeholder for time. If a signal is clear from its context then it will be denoted simply as $V$ (without the subscript) or, with as slight abuse of notation, as $V_t$. For a signal $Y_{\cdot}$, $Y_{0:t}=\{Y_{\tau}\}_{0 \leq \tau \leq t}$. If $\mathfrak{h}_1$ and $\mathfrak{h}_2$ are Hilbert spaces, $\mathscr{L}(\mathfrak{h}_1;\mathfrak{h}_2)$ denotes the class of all linear operators mapping from $\mathfrak{h}_1$ to $\mathfrak{h}_2$. If $\mathfrak{h}_1 =\mathfrak{h} =\mathfrak{h}_2$ then it is written simply as  $\mathscr{L}(\mathfrak{h})$. If $X,Y \in \mathscr{L}(\mathfrak{h})$ then $[X,Y]=XY-YX$ and $\{X,Y\} = XY + YX$.  If $X$ is an operator on the composite Hilbert space $\mathfrak{h}_1 \otimes \mathfrak{h}_2$ then $\mathrm{Tr}_{\mathfrak{h}_j}(X)$ denotes the partial trace of $X$ by tracing out over the Hilbert space $\mathfrak{h}_j$ ($j=1,2$). If $O_j \in \mathscr{L}(\mathfrak{h}_j)$ then $O_j$ is also used as a shorthand  for the ampliation of $O_j$ to the composite Hilbert space $\mathfrak{h}_1 \otimes \mathfrak{h}_2$. Also, $\delta_{jk}$ is  the Kronecker delta and 
the space of bounded operators in $\mathcal{L}(\mathfrak{h)}$ will be denoted by $B(\mathfrak{h})$. 

\section{Quantum non-Markovian model}
\label{sec:non-Markov-model}
We consider a principal quantum of interest on a finite dimensional Hilbert space $\mathfrak{h}_{\rm s}$ that is coupled to a compound environment consisting of an auxiliary quantum system on a finite-dimensional Hilbert space $\mathfrak{h}_{\rm a}$ which is coupled to $n_a$ environmental quantum  fields whose operators are labeled with the subscripts $k=1,\ldots,n_a$. The principal system is also separately coupled to a probe, whose operators are indexed by the subscript $0$. The boson Fock space \cite{KRP92} for each quantum field is denoted by $\Gamma_s(L^2(\mathbb{R}_+))$, where $\mathbb{R}_+ = [0,\infty)$ and $L^2(A)$ denotes the space of square integrable complex-valued functions over a Lebesque measurable set $A \subseteq \mathbb{R}_+^n$. The boson Fock space has the factorization property $\Gamma_s(L^2(\mathbb{R}_{+}))=\Gamma_{t]} \otimes \Gamma_{(t}$ \cite{KRP92}, where $\Gamma_{t]}=\Gamma_s(L^2([0,t]))$ and $\Gamma_{(t}=\Gamma_s(L^2((t,\infty)))$. The $n_{\rm a}+1$ quantum  fields of the environment and probe lives on the $(n_{\rm a}+1)$-fold tensor product of the boson Fock space  $\Gamma^{\otimes n_{\rm a}+1}_s(L^2(\mathbb{R}_+))=\otimes_{j=1}^{n_{\rm a}+1} \Gamma_s(L^2(\mathbb{R}_+))= \Gamma_s(L^2(\mathbb{R}^{n_{\rm a}+1}_+))$ and it has an analogous factorization property $\Gamma_s^{\otimes n_{\rm a}+1}(L^2(\mathbb{R}))=\Gamma_{t]}^{\otimes (n_{\rm a} + 1)} \otimes \Gamma_{(t}^{\otimes (n_{\rm a} + 1)}$. Let $\mathfrak{h}_{\rm sa} = \mathfrak{h}_{\rm s} \otimes \mathfrak{h}_{\rm a}$, $\mathcal{W}=B(\Gamma_s(L^2(\mathbb{R}_+)))$, $\mathcal{W}_{t]}=B(\Gamma_{t]})$ and $\mathcal{W}_{(t}=B(\Gamma_{(t})$.  A process $X_{\cdot}=\{X_t\}_{t \in \mathbb{R}_+}$ defined on $B(\mathfrak{h}_{\rm sa}) \otimes \mathcal{W}^{\otimes (n_a+1)}$ is said to be an {\em adapted process} if $X_t$ acts only on the factor $B(\mathfrak{h}_{\rm sa}) \otimes \mathcal{W}_{t]}^{\otimes (n_a+1)}$ and as the identity on the factor $\mathcal{W}_{(t}^{\otimes (n_a+1)}$.

The boson Fock space  for the individual quantum fields will be taken to be in the vacuum state $|\Omega\rangle_k$ $k=0,1,\ldots,n_{\rm a}$. For the probe there are the fundamental processes $A_{0,t}$, $A^{\dag}_{0,t}$ and $\Lambda_{0,t}$ while on the remaining $n_{\rm a}$ environmental quantum fields there are the fundamental processes $A_{k,t}$, $A^{\dag}_{k,t}$ and $\Lambda_{jk,t}$ for $j,k=1,\ldots,n_{\rm a}$. The forward differentials $dM_{k,t}=M_{k,t+dt} - M_{t}$, where $M$ can be any of the fundamental processes, satisfy the product \cite{KRP92,GJ09}, 
\begin{eqnarray*}
dA_{k,t}dA^{\dag}_{l,t} &=& \delta_{kl}dt,\; k=0,1,\ldots,n_{\rm a},\\
d\Lambda_{jk,t}d\Lambda_{lm,t} &=& \delta_{kl}d\Lambda_{jm,t}, \; j,k,l,m=1,\ldots,n_{\rm a},\\
d\Lambda_{0,t}dA^{\dag}_{0,t} &=& dA^{\dag}_{0,t}\\ 
d\Lambda_{jk,t}dA^{\dag}_{l,t} &=& \delta_{kl}dA^{\dag}_{j,t},\; j,k,l=1,\ldots,n_{\rm a},
\end{eqnarray*}
with all other products between $dA_{0,t}$, $dA^{\dag}_{0,t}$, $d\Lambda_{0,t}$,  $dA_{j,t}$, $dA^{\dag}_{k,t}$, and $d\Lambda_{lm,t}$, $j,k,l,m=1,\ldots,n_{\rm a}$, besides the above, vanishing.

The principal quantum system, auxiliary and quantum fields are coupled via the controlled QSDE introduced in \cite{BvH08}:
\begin{eqnarray}
dU_t &=& \left(-i\left(H_{\rm s}(t) +H_{\rm sa}(t) +\frac{1}{2}\sum_{k=0}^{n_a}L_k(t)^{\dag}L_k(t)\right) dt +  dA_{0,t}^{\dag} L_{0}(t)-L_0(t)^{\dag}S_0(t) dA_{0,t} \right. \notag \\
&&\quad + \sum_{k=1}^{n_a} dA_{k,t}^{\dag} L_{k}(t) - \sum_{j,k=1}^{n_a} L_j(t)^{\dag}S_{jk}(t) dA_{k,t} +(S_0(t)-I)d\Lambda_{0,t} \notag \\
&&\quad \left. + \sum_{j,k=1,\ldots,n_a} (S_{jk}(t) -\delta_{jk}) d\Lambda_{jk,t}\right)U_t,\; U_0 = I. \label{eq:controlled-QSDE}
\end{eqnarray}
In the above:
\begin{enumerate}
\item $H_{\rm s}(t)$, $L_0(t)$ and $S_0(t)$ are adapted principal and probe processes on $B(\mathfrak{h}_{\rm s})\otimes \mathcal{W}$, with an ampliation to $B(\mathfrak{h}_{\rm s})\otimes \mathcal{W}^{\otimes (n_{\rm a}+1)}$ that  acts as the identity operator on the auxiliary and all other environmental boson Fock spaces. For each $t$, $H_{\rm s}(t)$ is a Hamiltonian and is self-adjoint and $S_0(t)$ is a unitary operator. More specifications on these processes will be given below in \S  \ref{sec:controlled-flows}. 

\item For $j,k=1,\ldots,n_{\rm a}$ and each $t$, $H_{\rm sa}(t)$, $L_k(t)$ and $S_{jk}(t)$ are operators on $B(\mathfrak{h}_{\rm sa})$, with their ampliations acting as the identity operator on all boson Fock spaces. $H_{\rm sa}(t)$ is self-adjoint while the $S_{jk}(t)$'s form a unitary $n_{\rm a} \times n_{\rm a}$ matrix with operator entries. That is, letting
$$
\mathbf{S}(t) = \left[\begin{array}{cccc} S_{11}(t) & S_{12}(t) & \ldots & S_{1n_{\rm a}}(t) \\
S_{21}(t) & S_{22}(t) & \ldots & S_{2n_{\rm a}}(t)\\
\cdots & \vdots &  \ddots & \vdots \\
S_{n_{\rm a} 1}(t) & S_{n_{\rm a} 2}(t) & \ldots & S_{n_{\rm a} n_{\rm a}}(t) 
\end{array}\right], 
$$
then $\mathbf{S}(t)^{\dag}\mathbf{S}(t)=\mathbf{S}(t)\mathbf{S}(t)^{\dag}=I$.
Here $H_{\rm sa}(t)$ includes a coupling Hamiltonian term that couples the system and ancilla as well as the ancilla's own Hamiltonian, while the principal system's Hamiltonian can be included in $H_{\rm s}(t)$. 

\end{enumerate}

\section{Controlled quantum flow and output equation}
\label{sec:controlled-flows}
For any operator  $X$ on the principal quantum system and auxiliary, $j_t(X)=U_t^{\dag} X U_t$, where $U_t$ is a unitary solution to the controlled QSDE \eqref{eq:controlled-QSDE}, is called the {\em controlled quantum flow}. For the Markov case where there is no auxiliary but only coupling to a probe (in this case the field with index 0), the equation for the controlled quantum flow is given in \cite{BvH08}. The controlled quantum flow equation can be straightforwardly generalized to include  additional terms due to the auxiliary and environmental quantum fields that the principal is coupled to since $H_{\rm sa}(t)$, $L_k(t)$ and $S_{jk}(t)$ for $j,k=1,\ldots,n_{\rm a}$ are fixed operators in $B(\mathfrak{h}_{\rm sa})$ for each $t$ rather than (adapted) quantum processes, using the calculations that have been presented in \cite{GJ09}. The QSDE for the controlled quantum flow is:
\begin{eqnarray*}
dj_t(X) &=&  j_t\left(i[H_{\rm s}(t)+H_{\rm sa}(t),X] +\sum_{k=0}^{n_{\rm a}}\left(L_k(t)^{\dag}X L_{k}(t) -\frac{1}{2}\{L_k(t)^{\dag}L_{k}(t) ,X\} \right) \right)dt \\
&&\quad + dA_{0,t}^{\dag}j_t(S_{0}(t)^{\dag}[X,L_0(t)]) + j_t([L_0(t)^{\dag},X]S_{0}(t))dA_{0,t} \\
&&\quad  + j_t(S_0(t)XS_0(t)-X)d\Lambda_0(t)\\
&&\quad + \sum_{j,k=1,\ldots,n_{\rm a}} \left( dA_{j,t}^{\dag}j_t(S_{jk}(t)^{\dag}[X,L_k(t)]) + j_t([L_{j}(t)^{\dag},X] S_{jk}(t)) dA_{k,t}^{\dag}\right) \notag \\
&&\quad + \sum_{j=1}^{n_{\rm a}} \sum_{k,l=1}^{n_{\rm a}} j_t(S^{\dag}_{kj}(t) X S_{kl}(t) - X)d\Lambda_{lj,t}  
\end{eqnarray*}
for any operator $X \in B(\mathfrak{h}_{\rm sa})$. 

Throughout the paper we will consider a finite time interval $[0,T]$, with $0<T<\infty$. Following \cite{BvH08}, a very general class of measurements is performed on the principal quantum system through the probe, which includes adaptive measurements, based on a noise process $Z_{\cdot}$ on the probe of the form:
$$
Z_t =\int_{0}^t \Xi_{s} d\Lambda_{0,s} + \int_{0}^t \Upsilon_{s} dA^{\dag}_{0,s} + \int_{0}^t \Upsilon_{s}^{\dag} dA_{0,s},  
$$
such that $\Xi_{t}$, $\Upsilon_{t}$ are adapted processes that are affiliated to $\mathcal{Z}_t=\mathrm{vN}(Z_{0:t})$, the von Neumann algebra generated by $Z_{0:t}$ , $\Xi_t$ is self-adjoint and $\mathcal{Z}_t$ is a commutative algebra, for every $t \in [0,T]$. An example of such a noise process corresponding to adaptive homodyne measurements will be given in Example \ref{ex:adaptive-HD} below;  other examples can be found in \cite{BvH08}. 

The following assumption will be assumed to hold \cite{BvH08}.

\begin{assumption}
The processes $\Xi_{\cdot},\Upsilon_{\cdot},H_{\rm s}(\cdot),L_{0}(\cdot),S_{0}(\cdot)$ are bounded measurable processes, i.e.,
$\mathop{\sup}_{t \in [0,T]} \|\Xi_t\|< \infty$, $t \mapsto \Xi_t \psi$ is measurable for all $\psi \in \mathfrak{h}_{\rm sa} \otimes \Gamma_s(L^2(\mathbb{R}_+))$ and analogously for the other processes. 
\end{assumption}

The noise process $Z_{\cdot}$ of the probe becomes transformed to the output process $Y_t=j_t(Z_t)$ that is given by the QSDE \cite{BvH08}:
\begin{eqnarray}
dY_t &=&  dj_t(Z_{t}) \notag\\
&=& \Xi_{t} d\Lambda_{0,t} + j_t(S_{0}(t)^{\dag} (\Upsilon_{t} +\Xi_{t} L_{0}(t))dA_{0,t}^{\dag} +j_t((\Upsilon_{t}^{\dag} + \Xi_{t} L_{0}(t)^{\dag} )S_0(t))dA_{0,t} \notag \\
&&\quad + j_t(\Xi_{t}L_{0}(t)^{\dag}L_{0}(t) + \Upsilon_{t}^{\dag} L_{0}(t) + \Upsilon_{t}L_{0}(t)^{\dag})dt \label{eq:Y-QSDE}
\end{eqnarray}

The process $Y_{\cdot}$ satisfies the {\em self-non-demolition property}, $[Y_t,Y_s]=0$ for all $s,t \geq 0$. Note that this property also holds with the addition of the auxiliary and the environmental noise channels, as it  follows from the standard property that the solution $U_{\cdot}$ of the controlled QSDE is, like to usual Hudson-Parthasarathy QSDE, a unitary co-cycle with respect to the second quantization of the right shift operator on the 1-particle space of the probe boson Fock space; for a discussion, see \cite{KRP92}\cite[Chapter 2]{NY17}. By the same property the {\em non-demolition property} $[j_t(X),Y_s]=0$ for all $0 \leq s \leq t$ also holds for any operator $X$ of the principal system and auxiliary. 

For any operator $O$ that is affiliated to a commutative von Neumann algebra $\mathcal{C}$, 
we use $\iota$ to denote a map that takes $O$ and sends it to a classical random variable $\iota(O)$ that belongs to an algebra of essentially bounded random variables which is isomorphic to $\mathcal{C}$ via the Spectral Theorem \cite[Theorem 3.3]{BvHJ07}. For the remainder of this paper, for each fixed time $t$, $\mathcal{C}$ is implicitly taken to be the commutative von Neumann algebra generated by $\mathcal{Y}_t=\mathrm{vN}(Y_{0:t})$ and $\iota$ is a map that sends any operator $O$ affiliated to $\mathcal{Y}_t$ to a classical random variable $\iota(O)$ via the Spectral Theorem. We conclude this section with the following remark regarding further assumptions on the processes coupling the system to the probe. 

\begin{remark}
The processes $H_{\rm s}(\cdot)$, $L_{\rm 0}(\cdot)$, $S_{\rm 0}(\cdot)$ besides being adapted are taken to be processes that are affiliated to $B(\mathfrak{h}_{\rm s}) \otimes \mathcal{Z}_t$ for each $t$.
\end{remark}

\begin{example}
\label{ex:adaptive-HD} In adaptive homodyne detection the angle of the probe quantum field to be measured at time $t$ can be adapted by adjusting the relative phase between the  quantum field being measured and the local oscillator for the probe according to the measurement record up time before $t$. The noise process is modeled as:
$$
Z_t = \int_{0}^t \left(e^{i \theta(Z_{0:v})}dA^{\dag}_{0,v} + e^{-i \theta(Z_{0:v})}dA_{0,v}\right),
$$
where $\theta$ is a real essentially bounded Borel-measurable function. The output process $Y_{\cdot}$
follows from \eqref{eq:Y-QSDE} as:
\begin{align*}
Y_t &= \int_{0}^t \left(\vphantom{e^{-i\theta(Y_{0:t})}} (e^{-i\theta(Y_{0:v})}j_t(S_0(v)^{\dag}L_0(v)) + e^{i\theta(Y_{0:v})}j_t(S_0(v)L_0(v)^{\dag}))dv \right. \notag\\
&\qquad  \left. + e^{i \theta(Y_{0:v})}dA^{\dag}_{0,v} + e^{-i \theta(Y_{0:v})}dA_{0,v}\right).
\end{align*}

\end{example}
\section{Quantum filtering equations}
\label{sec:controlled-quantum-fitering}

Let the principal system and auxiliary be prepared in the initial state (density operator) $\rho_{\rm sa}$ and all the quantum fields indexed by $j=0,1,\ldots,n_{\rm a}$ be prepared in the vacuum state $|\Omega\rangle_j$. Define the state (in the operator algebra sense) $\mathbb{P}(X) = \mathrm{tr}(\rho_{\rm sa}|\Omega \rangle \langle \Omega| X)$, where $|\Omega \rangle = |\Omega \rangle_0 \otimes  \cdots \otimes |\Omega \rangle_{n_{\rm a}}$ for any operator  $X \in B(\mathfrak{h}_{\rm sa})\otimes \mathcal{W}^{\otimes (n_{\rm a}+1)}$. 

By the demolition property, the quantum conditional expectation of $j_t(X)$ given $\mathcal{Y}_{t}$ with respect to the state $\mathbb{P}$ exists and is well-defined \cite{BvHJ07}. Let this quantum conditional expectation be $\pi_t(X)= \mathbb{P}(j_t(X) \mid \mathcal{Y}_{t}$). It satisfies a QSDE that was derived in \cite{BvH08} for a principal quantum system that is only coupled to a probe. However, as mentioned above, since (i) the environmental fundamental processes $A_{k,t}$, $A_{k,t}^{\dag}$ and $\Lambda_{jk,t}$ for $j,k=1,\ldots,n_{\rm a}$ are  coupled to operators $L_k(t)$ and $S_{jk}(t)$ on $B(\mathfrak{h}_{\rm sa})$ for each $t$ rather than quantum stochastic processes (i.e., unlike $H_{\rm s}(t)$, $L_0(t)$, and  $S_0(t)$) and (ii) measurement is performed only on the probe, the derivation of the quantum filtering equation in \cite{BvH08} carries over {\em mutatis mutandis} with only the addition of some obvious terms. Thus we have following.

\begin{proposition}
\label{prop:innovations} (\cite[Proposition 4.1]{BvH08})
Define the innovations process
$$
I_t =Y_t -\int_{0}^t \pi_v(\Xi_v L_{0}(v)^{\dag}L_{0}(v) + \Upsilon_{v}^{\dag}L_{0}(v)+ \Upsilon_{v}L_{0}(v)^{\dag})dv.
$$
Then $I_t$ is a $\mathcal{Y}_t$-martingale. That is, $\mathbb{P}(I_t \mid \mathcal{Y}_s)=I_s$.
\end{proposition}

Recall the notation $\iota$ from Section \ref{sec:controlled-flows}. Let $\tilde{\Xi}_{\cdot}$ and $\tilde{\Upsilon}_{\cdot}$ be the classical scalar stochastic processes  $\tilde{\Xi}_{t} = \iota(j_t(\Xi_t))$ and $\tilde{\Upsilon}_{t} = \iota(j_t(\Xi_t))$. We have in the next proposition the quantum filtering equation when a measurement of the  amplitude quadrature is performed on $Y_{\cdot}$ (by continuous homodyne detection), corresponding to taking $\Xi_{t} =0$ $\forall t \geq 0$. 
\begin{proposition}
\label{prop:filtering-diffusion} (\cite[Proposition 4.2]{BvH08})
Suppose that $\Xi_t=0$ and $\Upsilon_t$ has a bounded inverse  $\forall t \in [0,T]$. Then $\pi_{\cdot}$ satisfies w.r.t the semimartingale observations $Y_{\cdot}$ the QSDE:
\begin{align*}
d\pi_t(X) &= \pi_t\left(i[H_{\rm s}(t)+H_{\rm sa}(t),X]  +\sum_{k=0}^{n_{\rm a}} \left(L_{k}(t)^{\dag}X L_k(t) -\frac{1}{2}\{L_k(t)^{\dag}L_{k}(t), X\}\right) \right)dt \\
&\quad +\left( \pi_{t}(\Upsilon_t^{-1}XL_{0}(t) + \Upsilon_t^{-1\dag}L_{0}(t)^{\dag}X)-\pi_t(\Upsilon_t^{-1}L_{0} (t)+ \Upsilon_t^{-1\dag}L_{0}(t)^{\dag})\pi_t(X)\right)dI_t,
\end{align*}
for all $X \in B(\mathfrak{h}_{\rm sa})$ with initial condition $\pi_0(X)=X$, where $I_{\cdot}$ is the innovations process in Proposition \ref{prop:innovations}.
\end{proposition}

In the case that $\Xi_t$ has a bounded inverse then a photon counting measurement is performed on $Y_{\cdot}$ and the quantum filtering equation is given by the next proposition.
\begin{proposition}
\label{prop:filtering-counting} (\cite[Proposition 4.3]{BvH08}) Suppose that $\Xi_t $ has a bounded inverse $\forall t \in [0,T]$. Then $\pi_{\cdot}$ satisfies w.r.t the semimartingale observations $Y_{\cdot}$ the QSDE:
\begin{align*}
d\pi_t(X) &= \pi_t\left(i[H_{\rm 0}(t)+H_{\rm sa}(t),X]  +\sum_{k=0}^{n_{\rm a}} \left(L_{k}(t)^{\dag}X L_k(t) -\frac{1}{2}\{L_k(t)^{\dag}L_{k}(t), X\}\right)\right)dt \\
&\quad + \left(\frac{\pi_{t-}((\Upsilon_{t-}+\Xi_{t-}L_{0}(t-))^{\dag}X(\Upsilon_{t-}+\Xi_{t-}L_{0}(t-))}{\pi_{t-}((\Upsilon_{t-}+\Xi_{t-}L_{0}(t-))^{\dag}(\Upsilon_{t-}+\Xi_{t-}L_{0}(t-))}-\pi_{t-}(X)\right)\tilde{\Xi}_{t-}^{-1}dI_t,
\end{align*}
for all $X$ with initial condition $\pi_0(X)=X$, where $I_{\cdot}$ is the innovations process in Proposition \ref{prop:innovations}.  
\end{proposition}

\begin{remark}
\label{rem:classical-notation} For notational simplicity and by a slight abuse of notation, we will use $Y_t$, $I_t$ and $\pi_t(X)$ to also denote their corresponding classical random variables $\iota(Y_t)$, $\iota(I_t)$ and $\iota(\pi_t(X))$. Whether the notation is referring to the original self-commuting operator-valued process or to the corresponding classical stochastic process (by the Spectral Theorem) can be inferred from the context. 
\end{remark}

\section{Main results}
\label{sec:main-results}
This section will present the main results that extend the non-Markovian SDEs obtained in \cite{Nurd25} to the case of non-Markovian quantum systems that undergo  continous measurement feedback by generalized continuous homodyne detection or photon counting.

There exists a stochastic density operator $\varrho_{\mathrm{sa},t}$ in $B(\mathfrak{h}_{\rm sa})$ such that $\pi_t(X)=\mathrm{tr}(\varrho_{\mathrm{sa},t} X)$ for all $t$. This density operator satisfies an operator-valued SDE referred to as a stochastic master equation (SME) \cite{BvHJ07,WM10} and the reduced stochastic state of the principal system only, denoted by $\varrho_{\mathrm{s},t}$ is obtained by taking the partial trace of $\varrho_{\mathrm{sa},t}$ over the ancilla Hilbert space $\mathfrak{h}_{\rm sa}$,
\begin{equation}
\varrho_{\mathrm{s},t} = \mathrm{Tr}_{\mathfrak{h}_{\rm sa}}(\varrho_{\mathrm{sa},t}). \label{eq:reduce-stochastic-principal}
\end{equation}

For the case of a quantum Markov model, \cite{BvH08} derives the SME and expresses it in terms of the matrix-valued processes $\tilde{L}_{0}(\cdot)$ and $\tilde{H}_{0}(\cdot)$, where $\tilde{L}_{0,\cdot}$ and $\tilde{H}_{\mathrm{s},\cdot}$ are $\mathbb{C}^{\mathrm{dim}(\mathcal{H}_{\mathrm{s}}) \times \mathrm{dim}(\mathcal{H}_0)}$-valued stochastic processes whose elements are obtained by applying $\iota$ elementwise to the matrix elements of $j_t(L_{0}(t))$ and $j_t(H_{\mathrm{s}}(t))$ for all $t$ \cite[Remark 4.1]{BvH08}. However, for the purposes of this paper, more structure is required and we introduce the following assumption:

\begin{assumption}
The $\mathcal{Z}_t$-adapted processes $H_{\rm s}(\cdot)$ and $L_{0}(\cdot)$ take the form $\sum_{j=1}^m K_j g_j(Z_{0:t})$ for each $t$, where $m \geq 1$ is some integer, $K_j$ is an operator in $B(\mathfrak{h}_{\rm s})$ and $g_j$ is an essentially bounded Borel measurable function for each $j$.  
\end{assumption}

Note that in the above assumption $K_j$ is an operator on the principal system (acts trivially on the ancilla). Then we have the following lemma. 

\begin{lemma}
Let $H_{\rm s}(t)=\sum_{k=1}^{m} H_{\mathrm{s},k}(t) f_k(Z_{0:t})$ and $L_{0}(t)= \sum_{k=1}^{n} L_{0,k}(t) g_k(Z_{0:t})$ for some operators $H_{\mathrm{s},k}(t)=H_{\mathrm{s},k}(t)^{\dag},L_{0,l}(t) \in B(\mathfrak{h}_{\rm s})$ and $f_k,g_l$ are some essentially bounded real and complex Borel measureable functionals  for $k=1,\ldots,m$, $l=1,\ldots,n$, respectively, and all $t$. Then the following identities hold:
\begin{eqnarray*}
\lefteqn{\pi_t([H_{\rm s}(t),X])}\\ &=& \sum_{k=1}^m f_k(Y_{0:t}) \mathrm{tr}([\varrho_{\mathrm{sa},t},H_{\mathrm{s},k}(t)]X)\\
\lefteqn{\pi_t\left(L_{0}(t)^{\dag}XL_0(t)-\frac{1}{2}\{L_0(t)^{\dag}L_0(t),X\} \right)}\\ 
&=&  \sum_{j,k=1}^n g_j(Y_{0:t})^{\dag} g_k(Y_{0:t}) \mathrm{Tr}\left(\left(\vphantom{\frac{1}{2}} L_{0,k}(t) \varrho_t L_{0,j}(t)^{\dag} -\frac{1}{2}\left\{\varrho_{\mathrm{sa},t}, L_{0,j}(t)^{\dag} L_{0,k}(t)  \right\} \right) X \right)\\
\lefteqn{\pi_t(\Upsilon_t^{-1}XL_{0}(t) + \Upsilon_t^{-1 \dag}L_{0}(t)^{\dag}X)}\\ &=& \sum_{k=1}^n \mathrm{Tr}\left(\left( g_k(Y_{0;t})\tilde{\Upsilon}_t^{-1} L_{0,k}(t) \varrho_{\mathrm{sa},t}  + \varrho_{\mathrm{sa},t} g_k(Y_{0;t})^{\dag} \tilde{\Upsilon}^{-1 \dag} L_{0,k}(t)^{\dag} \right)X\right),\\
\lefteqn{\pi_t\left((\Upsilon_{t} + \Xi_{t}L_{0}(t))^{\dag}X(\Upsilon_{t} + \Xi_{t}L_{0}(t))\right)}\\ 
&=& \mathrm{Tr}\left( \left(\tilde{\Upsilon}_{t} + \tilde{\Xi}_{t} \sum_{k=1}^m g_k(Y_{0:t})L_{0,k}(t)\right) \varrho_{\mathrm{sa},t}  \left(\tilde{\Upsilon}_{t} + \tilde{\Xi}_{t} \sum_{k=1}^m g_k(Y_{0:t})L_{0,k}(t)\right)^{\dag}X \right)
\end{eqnarray*}
\end{lemma}
\begin{proof}
From the definition of a quantum conditional expectation and the fact that $j_t(\Upsilon_t)$, $j_t(\Xi_t)$, $j_t(H_{\mathrm{s},k}(t))$ and $j_t(L_{0,l}(t))$ are all in the commutant of $\mathcal{Y}_t$ for all $k,l$, it follows by a straightforward calculation that: 
\begin{eqnarray*}
\lefteqn{\pi_t([H_{\rm s}(t),X])}\\ 
&=& \sum_{k=1}^m f_k(Y_{0:t}) \pi_t([H_{\mathrm{s},k}(t),X])\\
\lefteqn{\pi_t\left(L_{0}(t)^{\dag}XL_0(t)-\frac{1}{2}\{L_0(t)^{\dag}L_0(t),X\} \right)}\\ &=& \sum_{j,k=1}^n g_j(Y_{0:t})^{\dag} g_k(Y_{0:t})  \pi_t\left(\vphantom{\frac{1}{2}} L_{0,j}(t)^{\dag} XL_{0,k}(t)  - \frac{1}{2}\{L_{0,j}(t)^{\dag} L_{0,k}(t),X\}\right)\\
\lefteqn{\pi_t(\Upsilon_t^{-1}XL_{0}(t) + \Upsilon_t^{-1 \dag}L_{0}(t)^{\dag}X)}\\ &=& \sum_{k=1}^n \left( g_k(Y_{0;t})j_t(\Upsilon_t^{-1}) \pi_t(XL_{0,k}(t))  + g_k(Y_{0;t})^{\dag} j_t(\Upsilon^{-1})^{\dag}\pi_t(L_{0,k}(t)^{\dag}X)\right)  \\
\lefteqn{\pi_t((\Upsilon_{t} + \Xi_{t}L_{0}(t))^{\dag}X(\Upsilon_{t} + \Xi_{t}L_{0}(t))}\\ &=&  \pi_t\left(\left(\tilde{\Upsilon}_{t} + \tilde{\Xi}_t \sum_{k=1}^n g_k(Y_{0:t})L_{0,k}(t))\right)^{\dag}X\left(\tilde{\Upsilon}_{t} + \tilde{\Xi}_{t}\sum_{k=1}^n g_k(Y_{):t})L_{0,k}(t)\right)\right).
\end{eqnarray*}
The result follows by using the identity $\pi_t(X)=\mathrm{tr}(\varrho_{\mathrm{sa},t} X)$.
\end{proof}

Define the stochastic superoperators $\mathcal{L}(t)$ and $\mathcal{G}(t)$ as
\begin{eqnarray*}
\mathcal{L}(t) \varrho_{\mathrm{sa},t}&=& \left(  i\left[\varrho_{\mathrm{sa},t},H_{\rm sa}(t)+ \sum_{k=1}^m H_{\mathrm{s},k}(t)f_k(Y_{0:t})\right]  \right. \\
&&\quad + \sum_{j,k=1}^n g_j(Y_{0:t})^{\dag} g_k(Y_{0:t})  \left(\vphantom{\frac{1}{2}} L_{0,k}(t) \varrho_{\mathrm{sa},t} L_{0,j}(t)^{\dag}  -\frac{1}{2}\left\{\varrho_{\mathrm{sa},t}, L_{0,j}(t)^{\dag} L_{0,k}(t)  \right\} \right) \\
&&\quad + \left. \sum_{k=1}^{n_{\rm a}} \left(\vphantom{\frac{1}{2}} L_{k}(t) \varrho_{\mathrm{sa},t}  L_{k}(t)^{\dag}  -\frac{1}{2}\left\{\varrho_{\mathrm{sa},t}, L_{k}(t)^{\dag} L_{k}(t)  \right\} \right) \right)
\end{eqnarray*}
and
\begin{align}
\mathcal{G}(t) \varrho_{\mathrm{sa}.t}&= \sum_{k=1}^n  \left(g_k(Y_{0:t}) \tilde{\Upsilon}_t^{-1}  L_{0,k}(t) \varrho_{\mathrm{sa},t}  + \varrho_{\mathrm{sa},t}  g_k(Y_{0:t})^{\dag} \tilde{\Upsilon}^{-1 \dag} L_{0,k}(t)^{\dag} \right)
\end{align}

From Proposition \ref{prop:filtering-diffusion} we  have the following SDEs for the stochastic master equation in the diffusive case.
\begin{proposition}
\label{prop:SME-diffusion}
Suppose that $\Xi_t=0$ and $\Upsilon_t$ has a bounded inverse $\forall t \in [0,T]$. Then the conditional density operator $\varrho_{\mathrm{sa},\cdot}$ satisfies w.r.t the semimartingale observations $Y_{\cdot}$ the SDE:
\begin{eqnarray*}
d\varrho_{\mathrm{sa},t} &=& \mathcal{L}(t) \varrho_{\mathrm{sa},t} dt\\
&&\; +   \left( \mathcal{G}(t) \varrho_{\mathrm{sa},t}  -\varrho_{\mathrm{sa},t} \sum_{k=1}^n  \mathrm{Tr}\left((g_k(Y_{0:t})\tilde{\Upsilon}_t^{-1} L_{0,k}(t)  + g_k(Y_{0:t})^{\dag} \tilde{\Upsilon}^{-1 \dag}L_{0,k}(t)^{\dag})\varrho_{\mathrm{sa},t} \right)  \right)dI_t,
\end{eqnarray*}
where $I_{\cdot}$ is the innovations process in Proposition \ref{prop:innovations}.
\end{proposition}
\begin{proof}
The statement of the theorem is obtained by writing $ d\pi_t(X) =\mathrm{tr}(d\varrho_{\mathrm{sa},t} X)$.
\end{proof}

Analogously, for the pure jump case the SME follows from Proposition \ref{prop:filtering-counting} and is given by the following.
\begin{proposition}
\label{prop:jump-SME}
Suppose that $\Xi_t $ has a bounded inverse $\forall t \in [0,T]$. Then the conditional density operator $\varrho_{\mathrm{sa},t}$ satisfies w.r.t the semimartingale observations $Y_{\cdot}$ the SDE:
\label{prop:SME-counting}
\begin{eqnarray*}
d\varrho_{\mathrm{sa},t} &=& \mathcal{L}(t) \varrho_{\mathrm{sa},t} dt + \left(\frac{\mathcal{M}(t-)  \varrho_{\mathrm{sa},t-} \mathcal{M}(t-)^{\dag}}{\mathrm{Tr}(\varrho_{\mathrm{sa},t-}\mathcal{M}(t-)^{\dag}\mathcal{M}(t-))}-\varrho_{\mathrm{sa},t-}\right)\tilde{\Xi}_{t-}^{-1}dI_t,
\end{eqnarray*}
where $I_{\cdot}$ is the innovations process in Proposition \ref{prop:innovations} and  
$$
\mathcal{M}(t)= \tilde{\Upsilon}_{t} + \tilde{\Xi}_{t} \sum_{k=1}^n g_k(Y_{0:t})L_{0,k}(t)
$$  
\end{proposition}

Define a linear superprojection operator $\mathcal{P}:B(\mathfrak{h}_{\rm sa}) \rightarrow B(\mathfrak{h}_{\rm sa})$ with the properties \cite{Nurd25}:
\begin{enumerate}
\item $\mathcal{P}$ maps density operators to density operators.

\item $\mathrm{Tr}_{\mathfrak{h}_{\rm a}}(\mathcal{P} X) = \mathrm{Tr}_{\mathfrak{h}_{\rm a}}(X)$ for all $X \in B(\mathfrak{h}_{\rm sa})$.

\item $\mathcal{P} ((O \otimes  I_{\mathfrak{h}_{\rm a}}) X) = O  \mathcal{P} X$ for any operators $O \in B(\mathfrak{h}_{\rm s})$ and $X \in B(\mathfrak{h}_{\rm sa})$. 
\end{enumerate}
Let $\mathcal{Q}= \mathcal{I}-\mathcal{P}$, where $\mathcal{I}$ is the identity superoperator on $\mathscr{L}(\mathfrak{h}_{\rm sa})$. For any operator $Z \in \mathscr{L}(\mathfrak{h}_{\rm sa})$ define $Z^p=\mathcal{P}Z$ and  $Z^q=\mathcal{Q}Z$. Also, for any linear superoperator $\mathcal{X}:  \mathscr{L}(\mathfrak{h}_{\rm sa}) \rightarrow  \mathscr{L}(\mathfrak{h}_{\rm sa})$, define  $\mathcal{X}^{pp} = \mathcal{P} \mathcal{X} \mathcal{P}$,  $\mathcal{X}^{pq} = \mathcal{P} \mathcal{X} \mathcal{Q}$, $\mathcal{X}^{qp} = \mathcal{Q} \mathcal{X} \mathcal{P}$, and  $\mathcal{X}^{qq} = \mathcal{Q} \mathcal{X} \mathcal{Q}$. Then we have:
\begin{lemma}
\label{lem:projector-prop}
$\mathcal{P} (X(O \otimes  I_{\mathfrak{h}_{\rm a}}) ) =   (\mathcal{P} X) O$. Moreover, $[\mathcal{G}(t),\mathcal{P}]=[\mathcal{G}(t),\mathcal{Q}]=0$ and $[\mathcal{M}(t),\mathcal{P}]=[\mathcal{M}(t),\mathcal{Q}]=0$ for all $t$. 
\end{lemma}
\begin{proof}
Note that $\mathcal{P}(X^{\dag})=\mathcal{P}(X)^{\dag}$. Therefore, $\mathcal{P} (X(O \otimes  I_{\mathfrak{h}_{\rm a}}) )= \mathcal{P} ((O^{\dag} \otimes  I_{\mathfrak{h}_{\rm a}})X^{\dag})^{\dag}=(O^{\dag} \mathcal{P}X^{\dag})^{\dag}=(\mathcal{P}X)O$. From the definition of the superoperator $\mathcal{G}(t)$ given in \eqref{eq:gen-diffusive} and the fact that $\mathcal{P} L_{0,k}(t)\varrho_{\mathrm{sa},t} =L_{0,k}(t) \mathcal{P} \varrho_{\mathrm{sa},t}$ for $k=1,\ldots,n$ (taking $g(Y_{0:\cdot})$ as a classical stochastic process per Remark \ref{rem:classical-notation}), it immediately follows that  $[\mathcal{G}(t),\mathcal{P}]=0$ and therefore also $[\mathcal{G}(t),\mathcal{Q}]=0$ for all $t$. By a similar argument, it also holds that $[\mathcal{M}(t),\mathcal{P}]=[\mathcal{M}(t),\mathcal{Q}]=0$ for all $t$.
\end{proof}

\subsection{Stochastic quantum non-Markovian equation: The diffusive case}
Consider first the diffusive case. By Lemma \ref{lem:projector-prop} it follows in analogy to \cite[\S III]{Nurd25} that 
\begin{align}
d\varrho^{p}_{{\rm sa},t}&=(\mathcal{L}(t)^{pp} \varrho^{p}_{{\rm sa},t}+ \mathcal{L}(t)^{pq}  \varrho^{q}_{{\rm sa},t})dt \notag \\
&\quad + \left(\mathcal{G}(t) \varrho^{p}_{{\rm sa},t} - \varrho^{p}_{{\rm sa},t} \sum_{k=1}^n  \mathrm{Tr}\left((g_k(Y_{0:t})\tilde{\Upsilon}_t^{-1} L_{0,k}(t)  + g_k(Y_{0:t})^{\dag} \tilde{\Upsilon}^{-1 \dag}L_{0,k}(t)^{\dag})\varrho^p_{\mathrm{sa},t} \right)\right) dI_t, \label{eq:SDE-p}\\
d\varrho^{q}_{{\rm sa},t}&=(\mathcal{L}(t)^{qp} \varrho^{p}_{{\rm sa},t}+ \mathcal{L}(t)^{qq}  \varrho^{q}_{{\rm sa},t} )dt\notag \\
&\quad + \left(\mathcal{G}(t) \varrho^{q}_{{\rm sa},t} - \varrho^{q}_{{\rm sa},t} \sum_{k=1}^n  \mathrm{Tr}\left((g_k(Y_{0:t})\tilde{\Upsilon}_t^{-1} L_{0,k}(t)  + g_k(Y_{0:t})^{\dag} \tilde{\Upsilon}^{-1 \dag}L_{0,k}(t)^{\dag})\varrho^p_{\mathrm{sa},t}\right) \right) dI_t. \label{eq:SDE-q}
\end{align}
It follows that the SDE for $\varrho^{q}_{{\rm sa},\cdot}$ is linear for a given  $\varrho^{p}_{{\rm sa},\cdot}$ and can be rewritten as:
\begin{align}
d\varrho^{q}_{{\rm sa},t}&=d\mathcal{A}_{ \varrho^p_{{\rm sa},t}}(t) \varrho^{q}_{{\rm sa},t}  + \mathcal{L}(t)^{qp}  \varrho^{p}_{{\rm sa},t}dt \label{eq:SDE-q}
\end{align}
where $\mathcal{A}_{ \varrho^p_{{\rm sa},t}}(t)$ is a stochastic generator given by: 
\begin{eqnarray}
\lefteqn{d\mathcal{A}_{\varrho^p_{{\rm sa},t}}(t)} \notag \\  
&=&  \mathcal{L}(t)^{qq}dt  \notag \\
&&\; +\left( \mathcal{G}(t)-  \sum_{k=1}^n  \mathrm{tr}\left((g_k(Y_{0:t})\tilde{\Upsilon}_t^{-1} L_{0,k}(t)  + g_k(Y_{0:t})^{\dag} \tilde{\Upsilon}^{-1 \dag}L_{0,k}(t)^{\dag})\varrho^p_{\rm sa,t} \right) \mathcal{I} \right) dI_t, \label{eq:gen-diffusive}
\end{eqnarray}
and $\mathcal{I}$ is the identity operator on $\mathscr{L}(\mathfrak{h}_{\rm sa})$ as before. Set $\mathcal{A}_{ \varrho^p_{{\rm sa},t_0}}=0$ at an initial time $t_0$. 

Let $\Phi_{t,t_0}$ be a superoperator on $\mathscr{L}(\mathfrak{h}_{\rm sa})$ that is the solution to the SDE:
\begin{align}
d\Phi_{t,t_0} = d\mathcal{A}_{ \varrho^p_{{\rm sa},t}}(t) \Phi_{t,t_0},   \label{eq:SDE-Phi-diff}
\end{align}
with the initial condition $\Phi_{t_0,t_0} =\mathcal{I}$. Under the assumption that $\mathcal{A}_{ \varrho^p_{{\rm sa},\cdot}}(\cdot)$ is a matrix-valued semimartingale, the unique solution $\Phi_{t,t_0}$ is called the stochastic exponential of $\mathcal{A}_{ \varrho^p_{{\rm sa},\cdot}}(\cdot)$, which is invertible for each $t$ \cite{DY08}. Then the following result follows from the same proof as \cite[Lemma 1 and Theorem 1]{Nurd25} using \cite[Theorem 1.2]{DY08}.

\begin{theorem}
\label{thm:non-Markovian-diffusion}
Suppose that $\Xi_t=0$ and $\Upsilon_t$ has a bounded inverse  $\forall t \in [0,T]$. The matrix-valued stochastic process $\varrho^{p}_{{\rm sa},\cdot}$ satisfies the SDE
\begin{align}
d\varrho^{p}_{{\rm sa},t} &=\left(\vphantom{\int_0^t}\mathcal{L}(t)^{pp} \varrho^{p}_{{\rm sa},t}   + \mathcal{L}(t)^{pq} \Phi_{t,t_0} \varrho^{q}_{{\rm sa},t_0} + \int_{t_0}^t \mathcal{K}_{\rm s}(t,t')  \varrho^{p}_{{\rm sa},t'}dt' \right)dt \notag \\
&\quad + \left(\mathcal{G}(t) \varrho^{p}_{{\rm sa},t} - \varrho^{p}_{{\rm sa},t} \sum_{k=1}^n  \mathrm{Tr}\left((g_k(Y_{0:t})\tilde{\Upsilon}_t^{-1} L_{0,k}(t)  + g_k(Y_{0:t})^{\dag} \tilde{\Upsilon}^{-1 \dag}L_{0,k}(t)^{\dag})\varrho^p_{\mathrm{sa},t} \right)\right)dI_t, \label{eq:diffusion-reduced-SDE-p}
\end{align}
where $\Phi_{t,t_0}$ is the stochastic exponential solving the SDE \eqref{eq:SDE-Phi-diff} with a stochastic generator given by \eqref{eq:gen-diffusive}, and $\mathcal{K}_{\rm s}(t,t') = \mathcal{L}(t)^{pq} \Phi_{t,t_0} \Phi_{t',t_0}^{-1} \mathcal{L}(t')^{qp}$ is a two-time stochastic kernel. 
\end{theorem}

\subsection{Stochastic quantum non-Markovian equation: The pure jump case}
Turning our attention to the pure jump case, we have similarly to the diffusive case from Theorem \ref{prop:SME-counting} that
\begin{eqnarray*}
d\varrho^p_{\rm sa,t} &=&  (\mathcal{L}(t)^{pp} \varrho^{p}_{{\rm sa},t}+ \mathcal{L}(t)^{pq}  \varrho^{q}_{{\rm sa},t})dt \notag  + \left(\mathcal{N}_{\varrho^p_{\mathrm{sa},t-}}(t-) \varrho^p_{\mathrm{sa},t-}-\varrho^p_{\mathrm{sa},t-}\right)\tilde{\Xi}_{t-}^{-1} dI_t,\\
d\varrho^q_{\rm sa,t} &=& (\mathcal{L}(t)^{qp} \varrho^{p}_{{\rm sa},t}+ \mathcal{L}(t)^{qq}  \varrho^{q}_{{\rm sa},t})dt + \left(\mathcal{N}_{\varrho^p_{\mathrm{sa},t-}}(t-) \varrho^q_{\mathrm{sa},t-}-\varrho^q_{\mathrm{sa},t-}\right)\tilde{\Xi}_{t-}^{-1} dI_t,
\end{eqnarray*}
where the superoperator $\mathcal{N}_{\varrho^p_{\mathrm{sa},t}}(t)$  is given by:
$$
\mathcal{N}_{\varrho^p_{\mathrm{sa},t}}(t) \rho = \frac{\mathcal{M}(t) \rho \mathcal{M}(t)^{\dag}}{\mathrm{Tr}\left(\varrho^p_{\mathrm{sa},t}\mathcal{M}(t)^{\dag}\mathcal{M}(t)\right)}
$$
In this case the linear SDE for $\varrho^{q}_{{\rm sa},\cdot}$ is of the same form as \eqref{eq:SDE-q} but where the stochastic generator $\mathcal{A}_{ \varrho^p_{{\rm sa},t}}(t)$ is given by: 
\begin{eqnarray}
d\mathcal{A}_{ \varrho^p_{{\rm sa},t}}(t)  &=&  \mathcal{L}(t)^{qq}dt +( \mathcal{N}_{\varrho^p_{\mathrm{sa},t-}}(t-)-   \mathcal{I}) \tilde{\Xi}^{-1}_{t-} dI_t.
\label{eq:gen-jump}
\end{eqnarray}
Let $\Phi_{t,t_0}$ be a superoperator on $\mathscr{L}(\mathfrak{h}_{\rm sa})$ that is the solution to the SDE:
\begin{align}
d\Phi_{t,t_0} = d\mathcal{A}_{ \varrho^p_{{\rm sa},t}}(t) \Phi_{t-,t_0},   \label{eq:SDE-Phi-jump}
\end{align}
with the initial condition $\Phi_{t_0,t_0} =\mathcal{I}$. Under the assumption that $\mathcal{A}_{ \varrho^p_{{\rm sa},\cdot}}(\cdot)$ is a matrix-valued semimartingale such that $I + \Delta \mathcal{A}_{ \varrho^p_{{\rm sa},t}}(t)$ is invertible for all $s \in[0,T]$, where $\Delta \mathcal{A}_{ \varrho^p_{{\rm sa},t}}(t) = \mathcal{A}_{\varrho^p_{{\rm sa},t}}(t)-\mathcal{A}_{\varrho^p_{{\rm sa},t-}}(t-)$, the unique stochastic exponential $\Phi_{t,t_0}$  is invertible for each $t$ \cite{DY08}. The analog of Theorem \ref{thm:non-Markovian-diffusion} for the pure jump case is given by the following theorem. 
\begin{theorem}
\label{thm:non-Markovian-jump}
Suppose that $\Xi_t $ has a bounded inverse  and $I + \Delta \mathcal{A}_{ \varrho^p_{{\rm sa},s}}(t)$ is invertible for all $t \in[0,T]$. The matrix-valued stochastic process $\varrho^{p}_{{\rm sa},\cdot}$ satisfies the SDE
\begin{align}
d\varrho^{p}_{{\rm sa},t} &=\left(\vphantom{\int_0^t}\mathcal{L}(t)^{pp} \varrho^{p}_{{\rm sa},t}   + \mathcal{L}(t)^{pq} \Phi_{t,t_0} \varrho^{q}_{{\rm sa},t_0} + \int_{t_0}^t \mathcal{K}_{\rm s}(t,t')  \varrho^{p}_{{\rm sa},t'}dt' \right)dt \notag \\
&\quad + \left(\mathcal{N}_{\varrho^p_{\mathrm{sa},t-}}(t-)  -\mathcal{I} \right)\varrho^p_{\mathrm{sa},t-} \tilde{\Xi}_{t-}^{-1} dI_t \label{eq:jump-reduced-SDE-p}
\end{align}
where $\Phi_{t,t_0}$ is the stochastic exponential solving the SDE \eqref{eq:SDE-Phi-diff} with a stochastic generator given by \eqref{eq:gen-jump}, and $\mathcal{K}_{\rm s}(t,t') = \mathcal{L}(t)^{pq} \Phi_{t,t_0} \Phi_{t'-,t_0}^{-1} \mathcal{L}(t'-)^{qp}$ is a two-time stochastic kernel. 
\end{theorem}

\subsection{Discussion}
\label{sec:conclu}
As in the equations obtained in \cite{Nurd25}, it can be seen from the stochastic equations given in Theorems \ref{thm:non-Markovian-diffusion} 
and \ref{thm:non-Markovian-jump} 
there is a flow of past states $\varrho^p_{\mathrm{s},t'}$ to $\varrho^q_{\mathrm{sa},t}$ for $t'<t$ that eventually flows to $\varrho^p_{\mathrm{sa},t}$. This flow is captured in the second and third terms in the drift component of the non-Markovian SDE, which can be viewed as the output of a stochastic linear time-varying system with multiplicative noise as described by the SDE \eqref{eq:SDE-q} 
(but with different stochastic generators for the diffusive and pure jump cases) 
with the output 
\begin{equation}
z_t =\mathcal{L}(t)^{pq}\varrho^q_{\mathrm{sa},t}, \label{eq:z-equation}
\end{equation}
and initial condition $\varrho^q_{\mathrm{sa},t_0}$. Therefore, in terms of the output of this stochastic system we can write  \eqref{eq:diffusion-reduced-SDE-p} as:
\begin{align*}
d\varrho^{p}_{{\rm sa},t} &=\left(\vphantom{\int_0^t}\mathcal{L}(t)^{pp} \varrho^{p}_{{\rm sa},t}   + z_t \right) dt \notag \\
&\quad + \left(\mathcal{G}(t) \varrho^{p}_{{\rm sa},t} - \varrho^{p}_{{\rm sa},t} \sum_{k=1}^n  \mathrm{Tr}\left((g_k(Y_{0:t})\tilde{\Upsilon}_t^{-1} L_{0,k}(t)  + g_k(Y_{0:t})^{\dag} \tilde{\Upsilon}^{-1 \dag}L_{0,k}(t)^{\dag})\varrho^p_{\mathrm{sa},t} \right)\right)dI_t.
\end{align*}
and similarly we can write \eqref{eq:jump-reduced-SDE-p} as:
\begin{align*}
d\varrho^{p}_{{\rm sa},t} &=\left(\vphantom{\int_0^t}\mathcal{L}(t)^{pp} \varrho^{p}_{{\rm sa},t}   + z_t \right)dt+ \left(\mathcal{N}_{\varrho^p_{\mathrm{sa},t-}}(t-)  -\mathcal{I} \right)\varrho^p_{\mathrm{sa},t-} \tilde{\Xi}_{t-} dI_t.  
\end{align*}

That is, the flow of information from the $\mathcal{P}$ projection to the orthogonal $\mathcal{Q}$ projection and then back to the $\mathcal{P}$ projection is encapsulated in the signal $z$ defined in \eqref{eq:z-equation}. That is, $\varrho^q_{\mathrm{sa},t}$ acts a memory subsystem as depicted in Fig.~\ref{fig:memory}. Finally, recall that the reduced stochastic density operator $\varrho_{\mathrm{s},t}$ of the principal system can be obtained via \eqref{eq:reduce-stochastic-principal}. 

\begin{figure}[h!]
\begin{centering}
\includegraphics[scale=0.5]{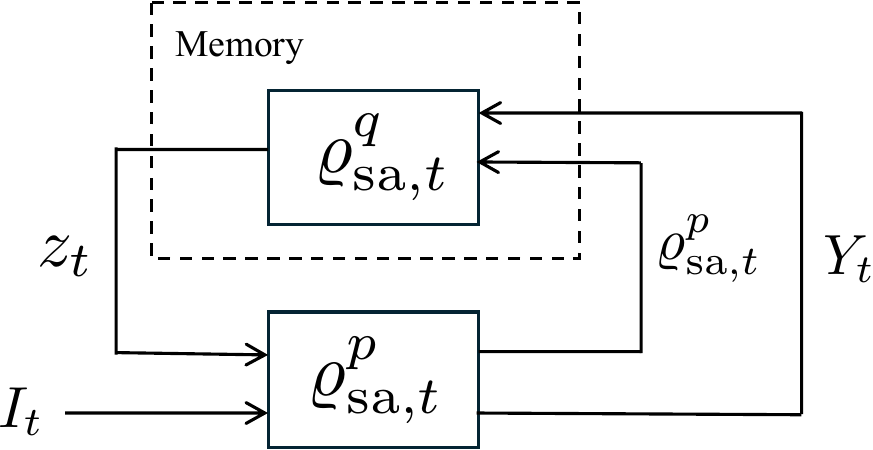}\caption{Feedback connection of the $\rho^p_{\mathrm{sa},\cdot}$ and $\rho^q_{\mathrm{sa},\cdot}$, with $\rho^q_{\mathrm{sa},\cdot}$ acting as the state of a memory subsystem (denoted by the dotted box).}\label{fig:memory}
\end{centering}
\end{figure}

This opens an approach to approximate the memory term, a problem posed in \cite{Nurd25}, by developing a suitable method to approximate $z$. 

\section{Conclusion and future work}
In this paper, we have generalized the non-Markovian stochastic SDEs for the evolution of a non-Markovian quantum system undergoing continuous measurements from \cite{Nurd25} to include measurement-feedback that is crucial to considering quantum feedback control of non-Markovian quantum systems. The derivation uses the framework of controlled QSDEs and quantum flows developed in \cite{BvH08}. Moreover, this paper derives the non-Markovian SDEs for the diffusive case  corresponding to  a generalized continuous homodyne measurement  and the pure jump case corresponding to a generalized continuous photon counting measurement, whereas \cite{Nurd25} only presented the equations for the diffusive case. The results of \cite{Nurd25} become a special case of the results of this paper.

Some open research problems that can be pursued following from this work and \cite{Nurd25} include developing a deeper theoretical understanding of the non-Markovian SDEs such as conditions for its physical realizability, i.e., the conditions on $\mathcal{L}^{pp}$ and the kernel $\mathcal{K}(t,t')$ so that $\rho^{p}_{\rm sa,t}$ is a valid density operator for all times $t$, model reduction of the SDE by approximating the signal $z$, system identification of the SDE from measurement data and the design of quantum feedback controllers for non-Markovian quantum systems. 

\bibliographystyle{IEEEtran}
\bibliography{notes-sp.bib}

\end{document}